\documentclass[preprint,12pt]{elsarticle}

\usepackage{amsmath, amsthm, amscd, amsfonts, amssymb, longtable, float}
\usepackage[utf8]{inputenc}

 \newtheorem{Theorem}{Theorem}
\newtheorem{Lemma}{Lemma}
\newtheorem{Example}{Example}
\newtheorem{Proposition}{Proposition}
\newtheorem{Definition}{Definition}

\newcommand{\dmp}{$d$-\textit{MP}}
\newcommand{\dmps}{$d$-\textit{MP} }
\newcommand{\dlmps}{$(d, \lambda)$-$MP$ }
\newcommand{\dlmp}{$(d, \lambda)$-$MP$}
\newcommand{\dff}{$d$-\textit{FFV}}
\newcommand{\dffs}{$d$-\textit{FFV} }
\newcommand{\rel}{$R_d$}

\journal{To be determined}

\begin{document}

\begin{frontmatter}



% \title{A novel algorithm to assess the reliability of multistate flow networks considering distance limitation}
\title{Assessing the reliability of multistate flow networks considering distance constraints}


\author[inst1]{Majid Forghani-elahabad}

\affiliation[inst1]{organization={Center of Mathematics, Computing, and Cognition, Federal University of ABC},%Department and Organization
            city={Santo André},
            postcode={09280-550}, 
            state={SP},
            country={Brazil}}

% \author[inst2]{Yi-Feng Niu}
% % \author[inst1,inst2]{Author Three}

% \affiliation[inst2]{organization={School of Traffic and Transportation, Beijing Jiaotong University},%Department and Organization
% %             addressline={Address Two}, 
%             city={Beijing},
%             postcode={100044}, 
% %             state={State Two},
%             country={China}}

\begin{abstract}
Evaluating the reliability of complex technical networks, such as those in energy distribution, logistics, and transportation systems, is of paramount importance. These networks are often represented as multistate flow networks (MFNs). While there has been considerable research on assessing MFN reliability, many studies still need to pay more attention to a critical factor: transmission distance constraints. These constraints are typical in real-world applications, such as Internet infrastructure, where controlling the distances between data centers, network nodes, and end-users is vital for ensuring low latency and efficient data transmission.
This paper addresses the evaluation of MFN reliability under distance constraints. Specifically, it focuses on determining the probability that a minimum of $d$ flow units can be transmitted successfully from a source node to a sink node, using only paths with lengths not exceeding a predefined distance limit of $\lambda $. We introduce an effective algorithm to tackle this challenge, provide a benchmark example to illustrate its application and analyze its computational complexity.
\end{abstract}


% \begin{highlights}
% \end{highlights}


\end{frontmatter}


\section{Introduction}
Stochastic or multistate flow network flow networks (MFNs) in which the arcs and maybe the nodes, and accordingly the entire network, can have more than two states have been desirable in the past decades due to their ability to model many real-world systems such as communication and telecommunication, logistic and transportation, power transmission and distribution, manufacturing systems and so forth~\cite{hao2020general, jane2017distribution, forghani2022improvedress, niu2020capacity, forghani2023usage, xu2022multistate, yeh2021novel, forghani2019iise, niu2017evaluating-MC, yeh2022rail, forghani2016amm, yeh2021reliability, forghani2019ijrqse, xiahou2023reliability, yeh2021novel, forghani2017ieee, nazarizadeh2022analytical}. 
For instance, in a manufacturing supply chain, nodes can represent suppliers or sources of raw materials, production or manufacturing facilities, distribution centers, warehouses, storage facilities, and markets or destinations. The arcs can represent the connections from suppliers to production facilities, from facilities to distribution centers or warehouses, and from distribution centers to customers or markets. In such a network, each node and arc can exist in various states. For instance, node states can vary depending on the availability of resources, production rate, machine reliability, etc. Arc states can vary due to factors such as road conditions, transformer availability, and their capacities.

Reliability indices are of great importance in evaluating the performance and quality of service of real-world systems, among which the $s$-$t$ terminal network reliability has been desirable in MFNs in recent decades. The basic definition of this index, denoted normally by $R_d$, is the probability of successfully transmitting a minimum of a given $d$ units of goods, commodities, or data from a source node to a destination node through the network. Engineers and designers can make informed decisions on network design, maintenance, and operation by accurately assessing the \(R_d\) to ensure reliable and efficient performance. Many exact and approximation algorithms have been proposed in the literature based on minimal paths (MPs)~\cite{forghani2023usage, forghani2019iise, forghani2017ieee, niu2020finding-MP, forghani2020assa, yeh2005novel, forghani2019ress, niu2021reliability, yeh2008simple, forghani2023improved, kozyra2023usefulness} or minimal cuts~\cite{forghani2016amm, forghani2019ijrqse, forghani2014ieee, kozyra2021innovative, forghani2019jcs, lamalem2020efficient, yeh2022application, forghani2021chapter, lin2009two, mansourzadeh2014comparative,huang2019reliability, forghani20223} to compute $R_d$. 
Chen et al.~\cite{chen2017search} employed the breadth-first search technique to introduce an enhanced MP-based algorithm that avoids generating duplicate solution vectors. 
Niu et al.~\cite{niu2020finding-MP} established a connection between reliability assessment and circulation problems, resulting in an efficient algorithm capable of solving feasible circulations for flow network reliability evaluation. 
The authors in~\cite{jane2008practical} introduced a novel direct decomposition algorithm for assessing network reliability without the prerequisite of having MPs or MCs in advance. 
In~\cite{forghani2023improved}, the authors delved into MP-based techniques and harnessed vectorization strategies to enhance the solution's performance. The authors then improved their algorithm by employing parallelization techniques and providing an efficient technique to avoid generating duplicate solutions~\cite{forghani2023usage}.  
Yeh~\cite{yeh2005novel} presented an innovative method for determining \dmp s, using a cycle-checking approach to validate each candidate. 
Bai et al.\cite{bai2016improved} explored the simultaneous determination of \dmp s for all possible demand levels $d$, proposing a recursive technique.
Niu et al.~\cite{niu2017new} tackled the issue by providing lower bounds on arc capacities and introducing a novel technique to identify duplicate solutions, thus presenting an efficient MC-based algorithm for problem-solving. 
The authors in~\cite{forghani2013ijor} reviewed several available results in the literature and showed that some results need to be revised in some cases. They then proposed an efficient MC-based algorithm for the reliability evaluation of MFNs that uses an improved approach for detecting duplicate solutions. The authors improved this algorithm in~\cite{forghani2016amm}.
Huang et al.~\cite{huang2019reliability} proposed a novel algorithm for the reliability evaluation of MFNs using a group approach that combines concepts from both MC and MP techniques. 

To enhance the practicality of the problem, several researchers have explored the integration of cost (budget) constraints into the reliability assessment of MFNs~\cite{niu2020capacity, forghani2019iise, forghani2019ijrqse, kozyra2023usefulness, lin2009two, niu2012reliability}.
Niu and Xu~\cite{niu2012reliability} contributed by proposing an enhanced algorithm for MFN reliability evaluation under cost constraints, incorporating a cycle-checking technique into their approach.
Considering budget constraints, Forghani-elahabad and Kagan~\cite{forghani2019iise} introduced an improved algorithm for MFN reliability assessment. Furthermore, they outlined how this algorithm could be effectively employed for evaluating the reliability of communication networks within smart grids. 
In a real-world application, \cite{lin2009two} modeled a manufacturing system, incorporating unreliable and multistate nodes and arcs with cost attributes. They presented an MP-based algorithm to assess the system's reliability while considering budget constraints. 
In~\cite{kozyra2023usefulness}, the author conducted a comprehensive review of MP- and MC-based approaches in the literature for calculating reliability under budget and maintenance cost constraints.
Moreover, some researchers have considered the time limitation on the flow transmission on a path resulting in the quickest path reliability problem and its extensions~\cite{forghani20223, forghani2015ieee, lin2003extend, forghani2015ress, yeh2009simpleQPR, forghani2022cnmac}.

In addition to time and cost considerations, the limitation on transmission distance plays a pivotal role in shaping the performance of an MFN. The distance between nodes within the network can exert a significant influence on the flow of goods and commodities, subsequently impacting the network's overall efficiency and reliability~\cite{cancela2011polynomial, zhang2018diameter}. To illustrate, let us consider a transportation network scenario where a delivery truck must cover extensive distances to reach its destination. As the distance increases, the likelihood of delays or breakdowns escalates, leading to a heightened risk of system failure and a decrease in the transportation network's performance. Thus, optimizing transmission distances stands as a crucial imperative in the design and management of MFNs across various real-world applications, encompassing transportation, communication, and distribution networks.

While the distance limitation is essential in several real-world systems such as Internet infrastructure, a few researchers have considered and studied this constraint in the reliability evaluation of MFNs~\cite{zhang2018diameter}. However, many researchers have investigated this issue in assessing the reliability of binary-state flow networks~\cite{cancela2011polynomial, canale2014monte, cancela2013monte, canale2015diameter}. In a binary-state network, the components, including nodes and arcs, and consequently the entire network, exhibit just two states. These two states often signify a binary choice, such as whether the component is open or closed~\cite{yeh2022application, cancela2011polynomial, yeh2022new}.
The authors in~\cite{cancela2011polynomial} proposed an innovative algorithm aimed at identifying and eliminating irrelevant arcs in the process of evaluating the reliability of a binary-state flow network when subjected to distance constraints. Studying the distance limitation in MFNs, the authors in~\cite{zhang2018diameter} have introduced the concept of irrelevant arcs in such networks and presented a strategy to remove them. Then, the authors proposed an approximation algorithm to evaluate the network reliability under distance constraints.

Expanding upon the foundational concept of MFN reliability, we define \(R_{(d,\lambda)}\)  as the reliability of an MFN with a demand level of \(d\) while subject to a distance constraint of \(\lambda\). This reliability measure quantifies the likelihood of successfully transmitting a minimum of \(d\) flow units from a source node to a target destination node, considering only the MPs with lengths not exceeding \(\lambda\).
A commonality in research on both binary-state flow networks MFNs is the identification of irrelevant arcs, defined as arcs that do not contribute to flow transmission due to distance constraints. Researchers have typically offered methods to detect and eliminate these irrelevant arcs, subsequently assessing the network's reliability post-removal. However, it is essential to acknowledge that these strategies often come with technical challenges in defining irrelevant arcs and the associated detection methods, not to mention their computationally expensive nature.
In this work, we introduce the notion of irrelevant MPs and propose an alternative approach that eliminates the need to identify or remove irrelevant arcs. Instead, our method directly sets the flow on irrelevant MPs to zero to satisfy the distance limitation. We demonstrate the correctness of our method through several results, compute its complexity results, and illustrate it through a known benchmark example.

We consider the following assumptions throughout this work.
\begin{enumerate}
    \item Each node is considered deterministic and exhibits perfect reliability.
    \item The capacity of each arc follows a random integer value distribution based on a predefined probability function.
    \item The capacities of individual arcs are statistically independent from one another.
    \item The network adheres to the flow conservation law \cite{ahuja1995network}. 
    \item All minimal paths are predetermined and provided in advance.
    \item Every arc within the network is part of at least one minimal path connecting node 1 to node \(n\).
\end{enumerate}

The remainder of this paper is structured as follows. 
Section~\ref{sec:network-modeling} introduces the necessary notations and terminology and lays out the foundational concepts for the problem's formulation.
Section~\ref{sec:algorithm} offers the primary findings and outlines an effective algorithm to tackle the problem.
An illustrative example and the complexity results are given in Section~\ref{sec:example-complexity}.
We conclude the work with the final remarks in Section~\ref{sec:conclusion}.

%% ======================================================================================

\section{Network modeling and problem formulation}\label{sec:network-modeling}

%%%%%%%%%%%%%%%%%%%%%%%%%%%%%%%%%%%%%%%%%%
\subsection{Notations and nomenclature}

Let us consider a multistate flow network (MFN), represented as $G(N$, $A$, $M$, $L)$. Here, $N = \{1, 2, \cdots, n\}$ signifies the collection of nodes and $A = \{a_1, a_2, \cdots, a_m\}$ represents the ensemble of arcs within the network. The vector $M = (M_1$, $M_2$, $\cdots, M_m)$ encompasses the maximum arcs' capacities, where $M_i$ denotes the utmost capacity of arc $a_i$, applicable to all $i = 1, 2, \cdots, m$. Correspondingly, the vector $L = (l_1, l_2, \cdots, l_m)$ serves as the length vector, with $l_i$ denoting the length of arc $a_i$, spanning $i = 1, 2, \cdots, m$. Notably, $n$ and $m$ denote the total number of nodes and arcs within the network, while nodes 1 and $n$ are expressly designated as the source and sink nodes, respectively. 
For instance, as depicted in Fig.~\ref{figure01}, the network comprises the set of nodes $N=\{$1, 2, 3, 4, 5$\}$ and the set of arcs  $ A = \{a_1,\cdots, a_8\}$.
One can consider  $M=($3, 2, 2, 1, 2, 1, 3, 2$)$ and $L=($1, 2, 1, 3, 2, 1, 2, 1$)$ respectively as the capacity and length vectors for this network. This way, for instance, $ M_5 = 2 $ and \(l_5=2\) signify respectively the maximum capacity and the length of arc $a_5$.

Consider a current system state vector (SSV) denoted as $X = (x_1$, $x_2$, $\cdots, x_m)$, where each component $x_i$ is a random integer ranging from 0 to $M_i$, representing the current capacity of arc $a_i$. For example, in the context of Fig.~\ref{figure01}, an SSV like $X = ($2, 2, 0, 0, 1, 0, 2, 0$)$ is valid, adhering to the constraint $0 \leq X \leq M$.
A path is defined as a sequence of adjacent arcs that connect node 1 to node $n$, while a minimal path (MP) is a path that does not contain any cycles. For instance, $P = {a_1, a_4, a_6, a_7}$ represents an MP in Fig.~\ref{figure01}. Let us denote all the MPs in the network as $P_1, P_2, \cdots, P_p$, so $p$ is the total number of MPs. For each $P_j$, let $LP_j$ represent its length and $CP_j(X)$ denote its capacity under the SSV,~$X$, for $j = 1, 2, \cdots, p$.
Hence, one notes that the maximum capacity of $P_j$ within the network $G(N, A, M, L)$ is $CP_j(M)$. The length of $P_j$ and its capacity under $X$ are computed as follows.
\begin{align}
& LP_j=\sum_{i:\ a_i\in P_j}l_i && \& &&   CP_j(X)=\min\{x_i|a_i\in P_j\}
\end{align}
For instance, we have \(LP_1 = l_1+l_4+ l_6 + l_7 = 1+3+1+2=7\) and $ CP_1(M) = \min\{M_1, M_4, M_6, M_7\} = \min\{3, 1, 1, 3\} $ for the network given in Fig.~\ref{figure01} with $M=($3, 2, 2, 1, 2, 1, 3, 2$)$ and $L=($1, 2, 1, 3, 2, 1, 2, 1$)$.
Let $e_i = (0, \cdots, 0, 1, 0, \cdots, 0)$ be an SSV in which the capacity level is 1	for $a_i$ and $0$ for the other arcs.
Let also $d$ be a non-negative integer number that represents the required flow to be transmitted from node 1 to node $n$ through the network, and $\lambda $ be a given transmission distance limit.
Assume that 
Finally, let $V(X)$ be the maximum flow of the network from node 1 to node $n$, and $e_i = (0,\ \cdots,\ 0,\ 1,\ 0,\ \cdots, 0)$ be an SSV in which the capacity level is 1 for $a_i$ and 0 for the other arcs. For instance, for the network of Fig.~\ref{figure01} and $X = ($2, 2, 0, 0, 1, 0, 2, 0$)$, we have $V(X)= 3$ and $V(X-e_2)=1$.
\begin{figure}[H]
\centering
\includegraphics[width=.25\textwidth]{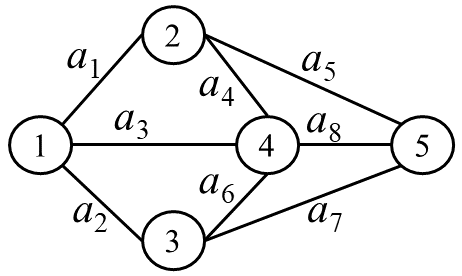}
\caption{A benchmark network example taken from~\cite{forghani2019ress}.}
\label{figure01}
\end{figure}

For any pair of system state vectors (SSVs), denoted as $X = (x_1$, $x_2$, $\cdots,\ x_m)$ and $Y = (y_1,\ y_2,\ \cdots,\ y_m)$, we define $X$ as less than or equal to $Y$, represented by $X\leq Y$, when $x_i\leq y_i$ holds for each $i = 1, 2, \cdots, m$. Additionally, we state that $X$ is strictly less than $Y$, indicated by $X<Y$, when $X\leq Y$ and there exists at least one index $j=1, 2, \cdots, m$ for which $x_j< y_j$.
A vector $X\in \Psi$ earns the label of \lq\lq minimal\rq\rq~if there is no other vector $Y\in \Psi$ such that $Y<X$. It is important to note that for a vector $X$ to be considered minimal within $\Psi$, $X$ does not need to be less than or equal to every other vector; instead, the crucial criterion is that no other vector in $\Psi$ is less than $X$.

%%%%%%%%%%%%%%%%%%%%%%%%%%%%%%%%%%%%%%%%%%
\subsection{Preliminaries}

Supposing that $d$ units of flow are transmitting from node 1 to node $n$ through the network, one notes that these flow should be divided among all the MPS~,$P_1,\ P_2,\ \cdots,\ P_p $. Letting $f_j$ denote the amount of flow being transmitted through $P_j$, we have a vector $F=(f_1,\ f_2,\ \cdots, f_p)$ which is called the feasible flow vector (FFV) at transmitting $d$ units of flow and denoted by \dff.
Let us begin with the case with no distance limitations.
To compute the reliability of an MFN in such a case, the concept of \dmps has been defined as follows in the literature~\cite{forghani2019iise, forghani2017ieee, niu2020finding-MP, forghani2019ress, lin1995reliability, yeh2007improved}.
\begin{Definition}
    A system state vector $X$ is a \dmps if and only if $V(X) = d$ and $V(X-e_i)=d-1$ for every $i$ with $x_i>0$.
\end{Definition}
\begin{Lemma}
    Assume that $\Psi(G, d) = \{X\ |\ 0\leq X\leq M\ \&\  V(X)\geq d\}$ and $\Psi_{\min}(G,d)$ is the set of all the minimal vectors in $\Psi(G, d)$. Then, $X$ is a \dmps if and only if $X\in \Psi_{\min}(G, d)$.
\begin{proof}
    If $X$ is a \dmp, then $V(X)=d$. On the other hand, for every $Y<X$, there exists an index $1\leq i\leq m$ such that $Y\leq X-e_i<X$,and hence $V(Y)\leq V(X-e_i)=d-1<d$. Accordingly, $Y\notin \Psi(G,d)$ signifying that $X$ is a minimal vector in $\Psi(G,d)$.
    Now, let $X\in \Psi_{\min}(G, d)$. If $V(X)> d$, then there exists an index $1\leq i\leq m$ such that $V(X-e_i)\geq d$ which contradicts $X$ being a minimal vector in $\Psi(G, d)$. Thus, $V(X)=d$. Similarly, one can shows that $V(X-e_i)<d$ for every $i$ with $x_i>0$, and thus $X$ is a \dmp.
\end{proof}
\end{Lemma}
As a result, if one finds all the \dmp s, say $\Psi_{\min}(G,d) = \{X^1,\ \cdots,\ X^{\sigma}\}$, and let $S_i = \{X\ | X^i\leq X\leq M\}$, then it is easy to see that $\Psi(G, d) = \cup_{i=1}^{\sigma}S_i$, and therefore the network reliability for the case with no distance limitations, denoted by \rel, can be computed by calculating a union probability.
Hence, the determination of all the \dmp s turns out to be an essential stage of indirect reliability evaluation approaches.

Many algorithms in the literature have used the following results to search for all the \dmp s in an MFN~\cite{jane2017distribution, forghani2019iise, forghani2017ieee, niu2020finding-MP, yeh2005novel, forghani2019ress, kozyra2023usefulness, forghani2022cnmac, lin2011using}.

\begin{Theorem}
    If $X=(x_1,\ \cdots,\ x_m)$ is a \dmp, then there exists a \dff, say $F=(f_1,\ \cdots,\ f_p)$, that satisfies the following systems.
    \begin{eqnarray}\label{th: dfftodmp}
        \begin{cases}%{ll}
        (i)\ f_1+f_2+\cdots + f_p=d, &\\
        (ii)\ 0\leq f_j\leq \min\{CP_j(M),d\}, & j=1,2,\cdots,p, \\
        (iii)\ \sum_{j:\ a_i\in P_j}f_j \leq {M_i}, & i=1,2,\cdots,m,\\
        (iv)\ x_i=\sum_{j:\ a_i\in P_j}f_j, & i=1,2,\cdots,m.
        \end{cases}
    \end{eqnarray}
\end{Theorem}

It is essential to recognize that while any \dmp, say $X$, corresponds to a \dff, say $F$, satisfying the system~\eqref{th: dfftodmp}, the reverse may not hold. In other words, not every $X$ obtained by solving the system above is necessarily a \dmp. Therefore, the available algorithms in the literature first find all the solutions of this system and then check each solution for being a \dmp ~\cite{jane2017distribution, forghani2019iise, forghani2017ieee, niu2020finding-MP, yeh2005novel, forghani2019ress, kozyra2023usefulness, forghani2022cnmac, lin1995reliability, lin2011using}. One of the most efficient techniques to check each solution of the system~\eqref{th: dfftodmp} to be a \dmps is the following lemma, which was presented in~\cite{yeh2005novel} and improved in~\cite{forghani2017ieee}.
\begin{Lemma}\label{lem:check}
    A solution $X$ of the system~\eqref{th: dfftodmp} is a \dmps if and only if no directed cycle exists in the network $G$ under $X$.
\end{Lemma}

Extending these results to the case with distance limitations and presenting the concept of \textit{irrelevant} MPs, we propose an efficient algorithm to compute the system reliability of an MFN under the distance limit.

%% ======================================================================================

\section{Searching for all the \dlmp s}\label{sec:algorithm}

\begin{Definition}
    The transmission distance for transferring $d$ units of flow from node 1 to node $n$ within the network $G$ is defined as the length of the longest MP through which at least one flow unit is being transmitted. 
\end{Definition}

One notes that there may exist different ways to transmit $d$ units of flow from node 1 to node $n$ through the network. However, once we are given the corresponding \dff, we know exactly how many flow units are being sent through each MP. To illustrate it, consider the network given in Fig.~\ref{fig:secondbench} with $M=(4, 3, 4, 5, 3, 4)$ and $L=(2, 1, 3, 2, 1, 2)$. It is easy to see that the maximum flow of the network from node 1 to node 4 is $V(M) = 9$ and that there are five MPs in the network; $P_1 = {a_1, a_4, a_6}$, $P_2 = {a_1, a_5}$, $P_3 = {a_2, a_4, a_5}$, $P_4 = {a_2, a_6}$, and $P_5 = {a_3}$. Now, to send four flow units from node 1 to node 4, we have 67 different scenarios, including $F_1= (0, 0, 0, 0, 4)$, $F_2=($0, 3, 0, 1, 0$)$, and $F_3 = ($4, 0, 0, 0, 0$)$. The associated transmission distance with $F_1, F_2$, and $F_3$ are respectively one, two, and three. Therefore, it is the corresponding \dff that determines the transmission distance. Therefore, we denote this parameter by $\Delta(G, d, F)$.

\begin{figure}
    \centering
    \includegraphics[width=0.3\linewidth]{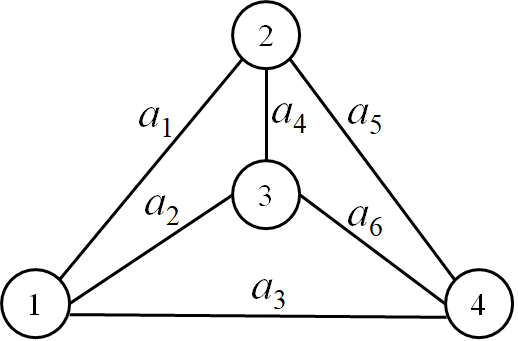}
    \caption{A simple benchmark example.}
    \label{fig:secondbench}
\end{figure}

\begin{Proposition}\label{prop:distance-cal}
    Assume that $d$ units of flow is transmitting from node 1 to node $n$ within the network $G$ and that $F=(f_1,\ f_2,\ \cdots, f_p)$ is its corresponding \dff. The transmission distance is equal to
    \begin{equation}
        \Delta(G, d, F) = \max\{LP_j\ |\ f_j>0,\ j = 1,\ 2,\ \cdots,\ p\}
    \end{equation}
\end{Proposition}
\begin{proof}
    It is concluded directly from the definition.
\end{proof}

\begin{Definition}
    A system state vector $X$ is a \dlmps candidate when there exists a \dff, say $F$, satisfying the system~\eqref{th: dfftodmp} and $\Delta(G, d, F) \leq \lambda $.
\end{Definition}
\begin{Definition}
    A \dlmps candidate, say $X$, is a (real) \dlmps if $V(X)=d$ and $V(X-e_i)<d$ for every $i$ with $x_i>0$.
\end{Definition}
\begin{Lemma}\label{lem:checkCan}
    A \dlmps candidate is a (real) \dlmps if no directed cycle exists in the network under it.
    \begin{proof}
        It can be concluded directly from the above definitions and Lemma~\ref{lem:check}.
    \end{proof}
\end{Lemma}

As explained earlier, even for a small benchmark network, there are numerous scenarios to transmit a given flow of magnitude $d$. It is important to note that the transmission distance varies depending on the specific \dffs being considered, each representing a unique scenario. Consequently, the task of finding all possible \dffs and their corresponding SSVs by solving the system~\eqref{th: dfftodmp}, subsequently verifying them for distance limitations and ensuring their status as (real) \dlmp, becomes a highly time-consuming endeavor.
To simplify this complexity, Zhang and Shao~\cite{zhang2018diameter} introduced the concept of irrelevant arcs in the context of MFNs. They classified these irrelevant arcs into two cases: (1) arcs not belonging to any MP and (2) arcs exclusive to MPs with lengths exceeding $\lambda $, which signifies the distance constraint. Subsequently, they proposed an approach to detect and eliminate these irrelevant arcs from the network. However, it is noteworthy to highlight that verifying all the MPs to identify every irrelevant arc is also a time-consuming process.

In light of Proposition~\ref{prop:distance-cal}, it becomes evident that the transmission distance depends on the specific \dffs and is determined via the lengths of MPs. Therefore, if one were to nullify the components of $F= ( f_1, \ \cdots,\ f_p)$ associated with MPs of lengths exceeding $\lambda $, the distance constraint would unquestionably be met by all computed \dffs. With this in mind, we introduce a novel concept: the \textit{irrelevant MP} (IMP), as defined below.

\begin{Definition}
    An MP is deemed irrelevant if its length exceeds the specified distance limitation. In other words, \(P_j\) is classified as irrelevant when \(LP_j=\sum_{i:\ a_i\in P_j}l_i> \lambda\).
\end{Definition}

Therefore, one can first solve the following system instead of the system~\eqref{th: dfftodmp} to calculate all the \dff s that satisfy the distance limitation.

\begin{eqnarray}\label{sys: dff}
        \begin{cases}%{ll}
        (i)\ f_j = 0, & \text{if } LP_j>\lambda\\
        (ii)\ f_1+f_2+\cdots + f_p=d, &\\
        (iii)\ 0\leq f_j\leq \min\{CP_j(M),d\}, & j=1,2,\cdots,p, \\
        (iv)\ \sum_{j:\ a_i\in P_j}f_j \leq {M_i}, & i=1,2,\cdots,m.
        \end{cases}
\end{eqnarray}
Then, the corresponding SSV to every calculated \dff can be computed as follows.

\begin{equation}\label{eq: dffTossv}
    x_i=\sum_{j:\ a_i\in P_j}f_j,\hspace{1cm} \forall i=1,2,\cdots,m.  
\end{equation}

\subsection{The proposed algorithm}
Now, we are ready to propose an efficient algorithm to address the problem of calculating all the \dlmp s in an MFN.\\

\noindent\textbf{Algorithm~1}

\textbf{Input:} \(G(N, A, M, L)\), its MPs, the demand level of \(d\), and the distance limitation of \(\lambda\).

\textbf{Output:} All the \dlmp s.

\noindent\textbf{Step~0.} Calculate \(LP_j=\sum_{i|a_i\in P_j}l_i\), for \(j=1, 2, \cdots, p\),  and compute $ KP_j (M)$ only for the MPs with $LP_j\leq \lambda $. Let \(S=\phi\).

\noindent\textbf{Step~1.} Find a solution $ F $ by solving the system~\eqref{sys: dff}. If no more solutions are found, stop.

\noindent\textbf{Step~2.} Calculate corresponding $X$ to $ F $ via Eq.~\eqref{eq: dffTossv}. 

\noindent\textbf{Step~3.} If there is a directed cycle in $G$ under $X$, then go to Step~1 to find the next solution.

\noindent \textbf{Step~4.} If $ X $ is not duplicated, add it to $ S $. \\

We should emphasize that steps 1 and 2 of the algorithm build upon the preceding discussion, while Step 3 relies on Lemma~\ref{lem:checkCan}. Furthermore, since it is feasible to generate the same  SSV corresponding to various \dff s, Step 4 of the algorithm eliminates duplicates. Consequently, we can formulate the following theorem.

\begin{Theorem}
Algorithm~1 calculates all the \dlmp s with no duplicates.
\end{Theorem}

%% ======================================================================================

\section{An illustrative example and the complexity results}\label{sec:example-complexity}

\subsection{An illustrative example}
A helpful method for gaining a deeper understanding of a novel approach is to examine its application through a simple benchmark network example.

\begin{Example}
Consider the depicted network in Fig.~\ref{figure01} with five nodes and eight arcs. Assume that the maximum capacity vector and the length vector are given respectively \(M=(3,2,2,1,2,1,3,2)\) and \(L=(1, 2, 1, 3, 2, 1, 2, 1)\). One needs to compute the reliability of this network for transmitting six units of flow from node 1 to node 5 with a transmission distance of not more than six. To solve this example, we need to find all the $(6, 6)$-$MP$, for which Algorithm~1 is utilized as follows.

\textit{Solution}: The demand level and distance limitation are respectively \(d=6\) and  \(\lambda=6\). We have nine MPs in this network: $ P_1=\{a_1, a_5\} $, $ P_2=\{a_2, a_7\} $, $ P_3=\{a_3, a_8\} $, $ P_4=\{a_1, a_4, a_8\} $,  $ P_5=\{a_2, a_6, a_2\} $,  $ P_6=\{a_3, a_4, a_5\} $, $ P_7=\{a_3, a_6, a_7\} $, $ P_8=\{a_1, a_4, a_6, a_7\} $, and $ P_9=\{a_2, a_6, a_4, a_5\} $.

\noindent \textbf{Step~0.} The MPs' lengths are calculated as follows. $LP_1 = 3$, 
$LP_2 = 4$, 
$LP_3 =2$, 
$LP_4 = 5$, 
$LP_5 = 4$, 
$LP_6 = 6$, 
$LP_7 = 4$, 
$LP_8 = 7$, 
and $LP_9 = 8$. 
The capacities of only the MPs with their length less than $\lambda $ are calculated as $KP_1(M) = $5, $KP_2(M) = $5, $KP_3(M) = $4, $KP_4(M) = $6, $KP_5(M) = $5, $KP_6(M) = $5, and $KP_7(M) = $6.

\noindent\textbf{Step~1.} The system~\eqref{sys: dff} has six solutions: (2,~2,~0,~1,~0,~0,~1,~0,~0), (1,~2,~1,~1,~0,~0,~1,~0,~0), (2,~2,~1,~0,~0,~0,~1,~0,~0), (2,~1,~1,~1,~0,~0,~1,~0,~0), (2,~2,~2,~0,~0,~0,~0,~0,~0), and (2,~2,~1,~1,~0,~0,~0,~0,~0).

\noindent\textbf{Step~2.} The corresponding SSVs to the calculated solutions in Step~2 are respectively: (3,~2,~1,~1,~2,~1,~3,~1), (2,~2,~2,~1,~1,~1,~3,~2), (2,~2,~2,~0,~2,~1,~3,~1), (3,~1,~2,~1,~2,~1,~2,~2), (2,~2,~2,~0,~2,~0,~2,~2), and  (3,~2,~1,~1,~2,~0,~2,~2).

\noindent\textbf{Step~3.} All the obtained SSVs meet the verification criteria, confirming that each is a $(6, 6)$-$MP$.

\noindent\textbf{Step~4.} There are no duplicates, meaning that the solution set includes all the computed vectors in Step~2.
\end{Example}

%%%%%%%%%%%%%%%%%%
\subsection{Complexity results}
We first recall that $m, n$, and $p$ are respectively the number of arcs, nodes, and MPs in the network. The time complexity of calculating the length or capacity of each MP is $O(m)$, and accordingly, Step~0 is of the order of $O(mp)$. System~\eqref{sys: dff} can be considered one constrained Diophantine equation and thus can be solved by the proposed approach of~\cite{forghani2015conbrazil}. Hence, the time complexity for calculating each solution in the worst case, and thus the time complexity of Step~1 is $O(p)$. The time complexity of Step~2 is $O(mp)$. The time complexity of Step~3 is $O(n)$~\cite{yeh2005novel}. We use the proposed approach of~\cite{forghani2023usage}, and thus, the time complexity of removing the duplicates in Step~4 is $O(m\sigma)$, where $\sigma$ is the number of obtained solutions including the duplicates. It is noteworthy that steps~0 and~4 are executed in parallel with other steps and that steps~1 to~3 are executed $\sigma$ times in the worst case. Therefore, the time complexity of Algorithm~1 is $O(mp+ \sigma(p+mp+n) + m\sigma) = O(mp\sigma)$, and we can formulate the following theorem.

\begin{Theorem}
    The time complexity of Algorithm~1 is $O(mp\sigma)$.
\end{Theorem}

%% ======================================================================================

\section{Conclusion}\label{sec:conclusion}
While distance limitations have been a subject of extensive research when it comes to evaluating the reliability of binary-state flow networks, this crucial aspect has received less attention in the context of multistate flow networks (MFNs). Given the significance of distance constraints in real-world systems like Internet infrastructure, this study dived into the problem of assessing the reliability of MFNs under these constraints.
In this work, we introduced the concept of irrelevant minimal paths (MPs). We presented an algorithm that sets the flow to zero on all relevant MPs to calculate all the \dlmp s in an MFN. We also outlined how the network's reliability can be calculated using these \dlmp s. To validate the proposed approach, we demonstrated its correctness and provided an illustration through a well-known benchmark example. Additionally, we offered insights into the computational complexity of the algorithm.

\section*{Acknowledgments}
The author thanks CNPq (grant 306940/2020-5) for supporting this work.

% \appendix

% \section{Sample Appendix Section}
% \label{sec:sample:appendix}
% Lorem ipsum dolor sit amet, consectetur adipiscing elit, sed do eiusmod tempor section \ref{sec:sample1} incididunt ut labore et dolore magna aliqua. Ut enim ad minim veniam, quis nostrud exercitation ullamco laboris nisi ut aliquip ex ea commodo consequat. Duis aute irure dolor in reprehenderit in voluptate velit esse cillum dolore eu fugiat nulla pariatur. Excepteur sint occaecat cupidatat non proident, sunt in culpa qui officia deserunt mollit anim id est laborum.

 \bibliographystyle{elsarticle-num} 
 % \bibliography{cas-refs}

\begin{thebibliography}{00}

\bibitem{hao2020general}
 Z. Hao, W. C. Yeh, Z. Liu, M. Forghani-elahabad, General multi-state rework network and reliability algorithm, Reliability Engineering \& System Safety 203 (2020) 107048.

\bibitem{jane2017distribution}
C. C. Jane, Y.-W. Laih, Distribution and reliability evaluation of maxflow in dynamic multi-state flow networks, European Journal of Operational Research 259 (3) (2017) 1045–1053.

\bibitem{forghani2022improvedress}
M. Forghani-elahabad, W. C. Yeh, An improved algorithm for reliability
evaluation of flow networks, Reliability Engineering \& System Safety 221
(2022) 108371.

\bibitem{niu2020capacity}
Y. F. Niu, X.-Z. Xu, C. He, D. Ding, Z.-Z. Liu, Capacity reliability
calculation and sensitivity analysis for a stochastic transport network,
IEEE Access 8 (2020) 133161–133169.

\bibitem{forghani2023usage}
M. Forghani-elahabad, E. Francesquini, Usage of task and data parallelism for finding the lower boundary vectors in a stochastic-flow network, Reliability Engineering \& System Safety (2023) 109417.

\bibitem{xu2022multistate}
 B. Xu, T. Liu, G. Bai, J. Tao, Y. Zhang, Y. Fang, A multistate network
approach for reliability evaluation of unmanned swarms by considering
information exchange capacity, Reliability Engineering \& System Safety
219 (2022) 108221.

\bibitem{yeh2021novel}
W. C. Yeh, Z. Hao, M. Forghani-elahabad, G.-G. Wang, Y.-L. Lin, Novel
binary-addition tree algorithm for reliability evaluation of acyclic multistate information networks, Reliability Engineering \& System Safety
210 (2021) 107427.

\bibitem{forghani2019iise}
M. Forghani-elahabad, N. Kagan, Reliability evaluation of a stochastic flow network in terms of minimal paths with budget constraint, IISE Transactions 51 (5) (2019) 547–558.

\bibitem{niu2017evaluating-MC}
Y. F. Niu, Z.-Y. G., W. H. K. Lam, Evaluating the reliability of a
stochastic distribution network in terms of minimal cuts, Transportation
Research Part E: Logistics and Transportation Review 100 (2017) 75–97.

\bibitem{yeh2022rail}
C. T. Yeh, Y.-K. Lin, L. C.-L. Yeng, P.-T. Huang, Rail transport network
reliability with train arrival delay: A reference indicator for a travel
agency in tour planning, Expert Systems with Applications 189 (2022)
116107.

\bibitem{forghani2016amm}
M. Forghani-elahabad, N. Mahdavi-Amiri, An improved algorithm for
finding all upper boundary points in a stochastic-flow network, Applied
Mathematical Modelling 40 (4) (2016) 3221–3229.

\bibitem{yeh2021reliability}
C. T. Yeh, Y.-K. Lin, L. C.-L. Yeng, P.-T. Huang, Reliability evaluation
of a multistate railway transportation network from the perspective of a
travel agent, Reliability Engineering \& System Safety 214 (2021) 107757.

\bibitem{forghani2019ijrqse}
M. Forghani-elahabad, N. Kagan, Assessing reliability of multistate flow
networks under cost constraint in terms of minimal cuts, International
Journal of Reliability, Quality and Safety Engineering 26 (05) (2019)
1950025.

\bibitem{xiahou2023reliability}
T. Xiahou, Y. X. Zheng, Y. Liu, H. Chen, Reliability modeling of modular k-out-of-n systems with functional dependency: A case study of
radar transmitter systems, Reliability Engineering \& System Safety 233
(2023) 109120.

\bibitem{forghani2017ieee}
M. Forghani-elahabad, L. H. Bonani, Finding all the lower boundary
points in a multistate two-terminal network, IEEE Transactions on Reliability 66 (3) (2017) 677–688.

\bibitem{nazarizadeh2022analytical}
F. Nazarizadeh, A. Alemtabriz, M. Zandieh, A. Raad, An analytical
model for reliability assessment of the rail system considering dependent
failures (case study of iranian railway), Reliability Engineering \& System
Safety 227 (2022) 108725.

\bibitem{niu2020finding-MP}
Y. F. Niu, X.-Y. Wan, X.-Z. Xu, D. Ding, Finding all multi-state minimal paths of a multi-state flow network via feasible circulations, Reliability Engineering \& System Safety 204 (2020) 107188.

\bibitem{forghani2020assa}
M. Forghani-elahabad, N. Mahdavi-Amiri, An algorithm to search for
all minimal cuts in a flow network, Advances in Systems Science and
Applications 20 (4) (2020) 1–10.


\bibitem{yeh2005novel}
W. C. Yeh, A novel method for the network reliability in terms of
capacitated-minimum-paths without knowing minimum-paths in advance, Journal of the Operational Research Society 56 (10) (2005) 1235–
1240.

\bibitem{forghani2019ress}
M. Forghani-elahabad, N. Kagan, N. Mahdavi-Amiri, An MP-based approximation algorithm on reliability evaluation of multistate flow networks, Reliability Engineering \& System Safety 191 (2019) 106566.

\bibitem{niu2021reliability}
 Y. F. Niu, C. He, D. Q. Fu, Reliability assessment of a multi-state
distribution network under cost and spoilage considerations, Annals of
Operations Research (2021) 1–20.

\bibitem{yeh2008simple}
W. C. Yeh, A simple minimal path method for estimating the weighted
multi-commodity multistate unreliable networks reliability, Reliability
Engineering \& System Safety 93 (1) (2008) 125–136.

\bibitem{forghani2023improved}
M. Forghani-elahabad, E. Francesquini, An improved vectorization algorithm to solve the d-mp problem, Trends in Computational and Applied
Mathematics 24 (1) (2023) 19–34.

\bibitem{kozyra2023usefulness}
P. M. Kozyra, The usefulness of (d, b)-mcs and (d, b)-mps in network
reliability evaluation under delivery or maintenance cost constraints,
Reliability Engineering \& System Safety 234 (2023) 109175.

\bibitem{forghani2014ieee}
 M. Forghani-elahabad, N. Mahdavi-Amiri, A new efficient approach to
search for all multi-state minimal cuts, IEEE Transactions on Reliability
63 (1) (2014) 154–166.

\bibitem{kozyra2021innovative}
P. M. Kozyra, An innovative and very efficient algorithm for searching
all multistate minimal cuts without duplicates, IEEE Transactions on
Reliability (2021).

\bibitem{forghani2019jcs}
 M. Forghani-elahabad, N. Kagan, An approximate approach for reliability evaluation of a multistate flow network in terms of minimal cuts,
Journal of Computational Science 33 (2019) 61–67.

\bibitem{lamalem2020efficient}
Y. Lamalem, K. Housni, S. Mbarki, An efficient method to find all d-MPs in multistate two-terminal networks, IEEE Access 8 (2020) 205618–
205624.

\bibitem{yeh2022application}
W. C. Yeh, C.-M. Du, S.-Y. Tan, M. Forghani-elahabad, Application of
lstm based on the bat-mcs for binary-state network approximated time-dependent reliability problems, Reliability Engineering \& System Safety
(2022) 108954.

\bibitem{forghani2021chapter}
M. Forghani-elahabad, 1 exact reliability evaluation of multistate flow
networks, in: Systems Reliability Engineering, De Gruyter, 2021, pp.
1–24.

\bibitem{lin2009two}
Y.-K. Lin, Two-commodity reliability evaluation of a stochastic-flow network with varying capacity weight in terms of minimal paths, Computers
\& Operations Research 36 (4) (2009) 1050–1063.

\bibitem{mansourzadeh2014comparative}
S. M. Mansourzadeh, S. H. Nasseri, M. Forghani-elahabad, A. Ebrahimnejad, A comparative study of different approaches for finding the upper boundary points in stochastic-flow networks, International Journal
of Enterprise Information Systems (IJEIS) 10 (3) (2014) 13–23.

\bibitem{huang2019reliability}
D.-H. Huang, C.-F. Huang, Y.-K. Lin, Reliability evaluation for a
stochastic flow network based on upper and lower boundary vectors,
Mathematics 7 (11) (2019) 1115.

\bibitem{forghani20223}
M. Forghani-Elahabad, 3 the disjoint minimal paths reliability problem,
in: Operations Research, CRC Press, 2022, pp. 35–66.

\bibitem{chen2017search}
X. Chen, J. Tao, G. Bai, Y. Zhang, Search for d-MPs without duplications in multistate two-terminal networks, in: 2017 Second International
Conference on Reliability Systems Engineering (ICRSE), IEEE, 2017,
pp. 1–7.

\bibitem{jane2008practical}
C. C. Jane, Y.-W. Laih, A practical algorithm for computing multi-state
two-terminal reliability, IEEE Transactions on reliability 57 (2) (2008)
295–302.

\bibitem{bai2016improved}
G. Bai, Z. Tian, M. J. Zuo, An improved algorithm for finding all minimal paths in a network, Reliability Engineering \& System Safety 150
(2016) 1–10.

\bibitem{niu2017new}
Y. F. Niu, Z. Y. Gao, W. H. K. Lam, A new efficient algorithm for finding
all d-minimal cuts in multi-state networks, Reliability Engineering \&
System Safety 166 (2017) 151–163.

\bibitem{forghani2013ijor}
M. Forghani-elahabad, N. Mahdavi-Amiri, On search for all d-mcs in
a network flow, Iranian Journal of Operations Research 4 (2) (2013)
108–126.

\bibitem{niu2012reliability}
Y.-F. Niu, X.-Z. Xu, Reliability evaluation of multi-state systems under
cost consideration, Applied Mathematical Modelling 36 (9) (2012) 4261–
4270.

\bibitem{forghani2015ieee}
M. Forghani-elahabad, N. Mahdavi-Amiri, A new algorithm for generating all minimal vectors for the q smps reliability problem with time
and budget constraints, IEEE Transactions on Reliability 65 (2) (2015)
828–842.

\bibitem{lin2003extend}
Y.-K. Lin, Extend the quickest path problem to the system reliability
evaluation for a stochastic-flow network, Computers \& Operations Research 30 (4) (2003) 567–575.

\bibitem{forghani2015ress}
M. Forghani-elahabad, N. Mahdavi-Amiri, An efficient algorithm for the
multi-state two separate minimal paths reliability problem with budget
constraint, Reliability Engineering \& System Safety 142 (2015) 472–481.

\bibitem{yeh2009simpleQPR}
W. C. Yeh, W.-W. Chang, C.-W. Chiu, A simple method for the multistate quickest path flow network reliability problem, in: 2009 8th International Conference on Reliability, Maintainability and Safety, IEEE, 2009, pp. 108–110.

\bibitem{forghani2022cnmac}
M. Forghani-elahabad, An improved algorithm for the quickest path
reliability problem, Proceeding Series of the Brazilian Society of Computational and Applied Mathematics 9 (1) (2022).

\bibitem{cancela2011polynomial}
H. Cancela, M. El Khadiri, L. A. Petingi, Polynomial-time topological
reductions that preserve the diameter-constrained reliability of a communication network, IEEE Transactions on Reliability 60 (4) (2011) 845–851.

\bibitem{zhang2018diameter}
Z. Zhang, F. Shao, A diameter-constrained approximation algorithm
of multistate two-terminal reliability, IEEE Transactions on Reliability
67 (3) (2018) 1249–1260.

\bibitem{canale2014monte}
E. Canale, F. Robledo, P. Romero, P. Sartor, Monte carlo methods in
diameter-constrained reliability, Optical Switching and Networking 14
(2014) 134–148.

\bibitem{cancela2013monte}
H. Cancela, F. Robledo, G. Rubino, P. Sartor, Monte carlo estimation
of diameter-constrained network reliability conditioned by pathsets and
cutsets, Computer Communications 36 (6) (2013) 611–620.

\bibitem{canale2015diameter}
E. Canale, H. Cancela, F. Robledo, P. Romero, P. Sartor, Diameter
constrained reliability: Complexity, distinguished topologies and asymptotic behavior, Networks 66 (4) (2015) 296–305.

\bibitem{yeh2022new}
W. C. Yeh, S.-Y. Tan, M. Forghani-elahabad, M. El Khadiri, Y. Jiang,
C.-S. Lin, New binary-addition tree algorithm for the all-multiterminal
binary-state network reliability problem, Reliability Engineering \& System Safety 224 (2022) 108557.

\bibitem{ahuja1995network}
R. Ahuja, T. Magnanti, J. Orlin, Network flows: theory, algorithms and
applications, Prentice Hall, 1995.

\bibitem{lin1995reliability}
J. S. Lin, C. C. Jane, J. Yuan, On reliability evaluation of a capacitated-flow network in terms of minimal pathsets, Networks 25 (3) (1995) 131–138.

\bibitem{yeh2007improved}
W. C. Yeh, An improved sum-of-disjoint-products technique for the symbolic network reliability analysis with known minimal paths, Reliability
Engineering \& System Safety 92 (2) (2007) 260–268.

\bibitem{lin2011using}
Y.-K. Lin, C. T. Yeh, Using minimal cuts to optimize network reliability for a stochastic computer network subject to assignment budget,
Computers \& Operations Research 38 (8) (2011) 1175–1187.

\bibitem{forghani2015conbrazil}
M. Forghani-elahabad, L. H. Bonani, Simple novel algorithm for a special diophantine system appeared in the system reliability problem, Proceeding Series of the Brazilian Society of Computational and Applied
Mathematics 3 (2) (2015).


\end{thebibliography}


\end{document}